\documentclass[letterpaper, 10 pt, conference]{ieeeconf}  

\usepackage[utf8]{inputenc}
\usepackage{amsfonts}
\usepackage{amsmath}
\usepackage{comment}
\usepackage{xcolor}
\usepackage{algorithm, algpseudocode}
\usepackage{subcaption}

\usepackage{cite}
\usepackage{hyperref}

\usepackage{pgfplots}
  \pgfplotsset{compat=newest}
  \usetikzlibrary{plotmarks}
  \usetikzlibrary{arrows.meta}
  \usepgfplotslibrary{patchplots}
  \usepackage{grffile}

\newcommand{\R}{\mathbb{R}}
\newcommand{\bmat}[1]{\begin{bmatrix}#1\end{bmatrix}}
\newcommand\norm[1]{\left\lVert#1\right\rVert}

\newtheorem{theorem}{Theorem}
\newtheorem{definition}{Definition}

\newtheorem{remark}{Remark}

\newtheorem{lemma}{Lemma}

\newtheorem{property}{Property} 

\IEEEoverridecommandlockouts                              
\title{\LARGE \bf
Imitation Learning with Stability and Safety Guarantees
}

\author{He Yin, Peter Seiler, Ming Jin and  Murat Arcak
\thanks{Funded in part by the Air Force Office of Scientific Research grant FA9550-18-1-0253, the Office of Naval Research grant N00014-18-1-2209, and  the U.S. National
	Science Foundation (NSF) grants ECCS-1906164. M. Jin acknowledges the funding from NSF EPCN-2034137.}
\thanks{H. Yin and M. Arcak are with the University of California, Berkeley {\tt\small \{he\_yin, arcak\}@berkeley.edu.}}
\thanks{P. Seiler is with the University of Michigan,  Ann Arbor {\tt\small pseiler@umich.edu.}}
\thanks{M. Jin is with the Virginia Tech {\tt\small jinming@vt.edu.}}
}

\begin{document}

\renewcommand*{\thefootnote}{\fnsymbol{footnote}}

\maketitle
\thispagestyle{empty}
\pagestyle{empty}
\setcounter{footnote}{1}

\begin{abstract}
  A method is presented to learn neural network (NN) controllers with stability and safety guarantees through imitation learning (IL). Convex stability and safety conditions are derived for linear time-invariant plant dynamics with NN controllers by merging Lyapunov theory with local quadratic constraints to bound the nonlinear activation functions in the
  NN. These conditions are incorporated in the IL process, which  minimizes the IL loss, and maximizes the volume of the region of attraction associated with the NN controller simultaneously. An alternating direction method of
multipliers based algorithm is proposed to solve the IL problem. The  method is illustrated on an inverted pendulum system, aircraft longitudinal dynamics, and vehicle lateral dynamics.
\end{abstract}


\section{Introduction}
    Imitation learning (IL) is a class of methods that learns a policy to attain a control goal consistent with expert demonstrations \cite{osa2018algorithmic, ALVINN1989}. Used in tandem with deep neural networks (NNs), IL presents unique advantages, including a substantial increase in sample efficiency compared to reinforcement learning (RL) \cite{sun2017deeply}, and wide applicability to domains where the reward model is not accessible or on-policy data is difficult/unsafe to acquire~\cite{osa2018algorithmic}. While IL is closely related to supervised learning as it trains a mapping from observations to actions \cite{daume2009search}, a key difference is the ensuing deployment of the learned policy under dynamics, which consequently raises the issue of closed-loop stability. This problem naturally falls within the realm of robust control, which analyzes stability for uncertain linear or nonlinear systems; however, a major technical challenge is to derive nonconservative guarantees for highly nonlinear policies such as NNs that can be also tractably incorporated into the learning process. 
    
    This letter tackles this issue and presents a method to learn NN controllers with stability and safety guarantees through IL. We first derive convex stability and safety conditions for linear time-invariant (LTI) plant dynamics by merging Lyapunov theory with local sector quadratic constraints (QCs) to describe the activation functions in the NN. Then we incorporate these constraints in the IL process that minimizes the IL loss, and maximizes the volume of
    the region of attraction associated with the NN controller. Finally, we propose an alternating direction method of multipliers (ADMM) based method to solve the IL problem.
    
    
    While there are many works focusing on NN robustness analysis using QCs, e.g., NN safety verification \cite{2019Fazlyab}, NN Lipschitz constant estimation \cite{FazlyabLip}, reachability analysis \cite{Haimin2020} and stability analysis \cite{Yinstabanaly, pauli2020offset, jin2018stability} of NN controlled systems, there are relatively fewer results regarding NN synthesis, including NN training with bounded Lipschitz constants \cite{Pauli2020}, and recurrent NN training with bounded incremental $\ell_2$ gains~\cite{revay2020convex}.
     
     Compared to existing works, this letter makes the following contributions. First, it presents a safe IL algorithm to synthesize NN controllers, which trades off between IL accuracy and the size of the stability margin of the NN controller. To the best of our knowledge, this is the first attempt to ensure stability of IL based on deep NNs. 
     Second, the stability condition from \cite{Yinstabanaly} is nonconvex and thus computationally intractable for NN control synthesis; here we convexify this constraint (using loop transformation) for its efficient enforcement in the learning process. 
     Third, as demonstrated by the numerical example in Section~\ref{sec:GTM_ex}, while the proposed approach can train a policy that imitates the expert demonstrations, it can  potentially improve local stability over  suboptimal expert policies, thus enhance the robustness of IL.
     
     Other related works include approximating explicit model predictive control law using NNs for linear parameter varying plant dynamics \cite{MomoIL}, and LTI plant dynamics \cite{karg2020efficient, SChen2018} to expedite the online evaluation of controllers. Compared with these works, the proposed method provides stability guarantees for the learned NN controller.

    The stability and safety issues for learning (especially RL) based control has also been addressed by a Hamilton-Jacobi reachability-based framework in \cite{fisac2018general}, a control barrier function based method in \cite{taylor2020control,lindemann2020learning}, and a convex projection operator based method in \cite{donti2020enforcing}.
     

\emph{Notation:} $\mathbb{S}^n$, $\mathbb{S}_+^n$ and $\mathbb{S}_{++}^n$ denote
the sets of $n$-by-$n$ symmetric, positive semidefinite and positive
definite matrices, respectively.  
When applied to vectors, the orders $>, \leq$ are applied elementwise.
For $P \in \mathbb{S}_{++}^{n}$,
 define the ellipsoid
\begin{align} \label{eq:epsil_def}
\mathcal{E}(P):= \{x \in \R^n : x^\top P x \leq
1\}. 
\end{align}


\section{Problem Formulation}\label{sec:problem_setup}

Consider the feedback system consisting of a plant $G$ and
state-feedback controller $\pi$ as shown in
Figure~\ref{fig:NominalFeedback}. We assume the plant $G$ is a linear,
time-invariant (LTI) system defined by the following discrete-time model:
\begin{align}
\label{eq:NominalSys}
x(k+1) &= A_G\ x(k) + B_{G} \ u(k),
\end{align} 
where $x(k) \in \R^{n_G}$ is the state and $u(k) \in \R^{n_u}$ is the control, $A_G \in \R^{n_G \times n_G}$ and $B_G \in \R^{n_G \times n_u}$.  Finally, assume $x(k)$ is constrained to a set $X \subset \R^{n_G}$, which is referred to as the ``safety condition''. This state constraint set is assumed to be a polytope symmetric around the origin:
\begin{align}
    X = \{x \in \R^{n_G}: -h \leq H x \leq h, \ h \ge 0\},
\end{align}
where $H \in \R^{n_X \times n_G}$, and $h \in \R^{n_X}$.
The controller $\pi: \R^{n_G} \rightarrow \R^{n_u}$ is an $\ell$-layer,
feedforward neural network (NN) defined as:
\begin{subequations}\label{eq:NNlong}
  \begingroup
\allowdisplaybreaks
\begin{align}
w^0(k) &= x(k), \label{eq:def_w0}\\
\label{eq:NNlong_wi}
w^i(k) &= \phi^i\left(\ W^i w^{i-1}(k) + b^i \ \right), 
           \ i = 1, \ldots , \ell, \\
u(k) &= W^{\ell+1} w^{\ell}(k) + b^{\ell+1}, \label{eq:output_layer}
\end{align}
\endgroup
\end{subequations}
where $w^i \in \R^{n_i}$ are the outputs (activations) from the
$i^{th}$ layer and $n_0 = n_G$.  The operations for each layer are defined by a weight
matrix $W^i \in \R^{n_i \times n_{i-1}}$, bias vector
$b^i \in \R^{n_i}$, and activation function
$\phi^{i}: \R^{n_i} \rightarrow \R^{n_i}$. The activation function
$\phi^i$ is applied element-wise, i.e.
\begin{align}
\phi^i(v) := \bmat{\varphi(v_1), \cdots, \varphi(v_{n_i})}^\top,
\end{align}
where $\varphi: \R \rightarrow \R$ is the (scalar) activation function
selected for the NN. Common choices for the scalar
activation function include  $\tanh$, ReLU, and leaky ReLU. We assume the  activation $\varphi$ is identical in all layers; this can be relaxed with minor changes to the notation. We also assume that the activation $\varphi$ satisfies $\varphi(0) = 0$.

Our goal is to stabilize an equilibrium, which we assume is shifted to the origin as is standard in control design. To ensure that the equilibrium remains at the origin with the NN controller, we impose the constraint $\pi(0) = 0$. Since $\pi(0) = 0$ translates to a nonconvex constraint on $(W^i, b^i)$, we set all the bias terms to be zero: $b^i
	= 0_{n_i \times 1}$, for $i = 1,..., \ell+1$.

\begin{remark}	\label{remark:zero_bias}
Setting the bias terms to zero is arguably an underuse of NNs.
Whether $\pi(0)=0$ can be achieved with less restrictive convex constraints
is an interesting problem that deserves further research.
\end{remark}

\vspace{-0.1in}
\begin{figure}[h]
	\centering
	\includegraphics[width=0.4\textwidth]{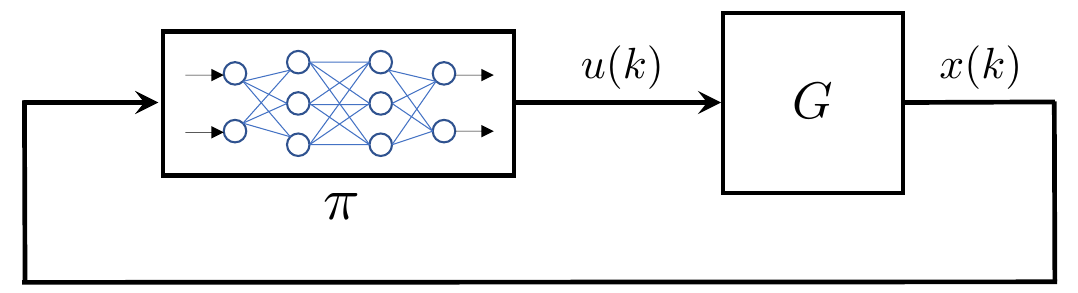}
	\caption{Feedback system with plant $G$ and NN $\pi$}
	\label{fig:NominalFeedback}    
\end{figure}
\vspace{-0.1in}

If control constraint sets are hypercubes, they can be considered in the framework by applying activation functions, like $\tanh$, elementwise to the output layer~\eqref{eq:output_layer}.

Let $\chi(k;x_0)$ denote the solution to the feedback system at time $k$
from initial condition $x(0)=x_0$. The region of attraction (ROA) is defined below. 
\begin{definition}
  \label{def:ROA}
  The region of attraction (ROA) of the feedback system with plant $G$ and NN $\pi$ is defined
  as
  \begin{align}
    \mathcal{R} := \{x_0 \in X: \lim_{k \rightarrow \infty} \chi(k;x_0) = 0_{n_G \times 1}\}.
  \end{align}
\end{definition}

Given state and control data pairs from the expert demonstrations, our goal is to learn a NN controller from the data to reproduce the demonstrated behavior, while guaranteeing the system trajectories under the control of the learned NN controller satisfy the safety condition $(x(k) \in X \ \forall k \ge 0)$, and are able to converge to the equilibrium point if they start from the ROA associated with the learned NN controller.


\section{Stability and Safety Conditions for NN Controlled Systems} \label{sec:nomanaly}
In this section, we treat the parameters of the NN controller as fixed and derive the safety and local stability conditions of the NN-controlled LTI systems.

\vspace{-0.1cm}
\subsection{NN Representation: Isolation of Nonlinearities}

It is useful to isolate the nonlinear activation functions from the
linear operations of the NN as done in \cite{2019Fazlyab,2018Kim}.
Define $v^i$ as the input to the activation function $\phi^i$ (recalling that $b^i=0_{n_i \times 1}$ by Remark~\ref{remark:zero_bias}):
\begin{align}
\label{eq:PhiInput}
  v^i(k) := W^i w^{i-1}(k),
   \ i = 1, \ldots , \ell.
\end{align}
The nonlinear operation of the $i^{th}$ layer
\eqref{eq:NNlong_wi} is thus expressed as
$w^i(k) = \phi^i( v^i(k) )$. Gather the inputs and
outputs of all activation functions:
\begin{align}
  v_\phi := \left[\begin{smallmatrix}v^1 \\ \vdots \\ v^{\ell}\end{smallmatrix}\right] \in \R^{n_\phi} 
\mbox{ and }
  w_\phi := \left[\begin{smallmatrix} w^1 \\ \vdots \\ w^{\ell}\end{smallmatrix}\right] \in \R^{n_\phi},
\end{align}
where $n_\phi := \sum_{i=1}^\ell n_i$, and define the combined
nonlinearity $\phi: \R^{n_\phi} \rightarrow \R^{n_\phi}$ by stacking
the activation functions:  
\begin{align}
    \phi(v_\phi) := \left[\begin{smallmatrix} \phi^1(v^1) \\ \vdots \\ \phi^\ell(v^\ell)\end{smallmatrix}\right]. 
\label{eq:phi_def}
\end{align}
Thus
$w_\phi(k) = \phi( v_\phi(k) )$, where the scalar activation function
$\varphi$ is applied element-wise to each entry of $v_\phi$.
Finally, the NN control policy $\pi$ defined in \eqref{eq:NNlong} can be rewritten as:
\begin{subequations}\label{eq:NNlft}
\begin{align}
  \left[\begin{smallmatrix}u(k)\\ v_\phi(k)\end{smallmatrix}\right] & = N  \left[\begin{smallmatrix}x(k) \\ w_\phi(k)\end{smallmatrix}\right] \label{eq:def_N}\\
  w_\phi(k) & = \phi( v_\phi(k) ).
\end{align}  
\end{subequations}
The matrix $N$ depends on the weights as follows, where the vertical and horizontal bars  partition $N$ 
compatibly with the inputs $(x,w_\phi)$ and outputs $(u,v_\phi)$:
\begin{align*}
N & = {\tiny\left[ \begin{array}{c|cccc} 
  0 & 0 &  0 & \cdots & W^{\ell+1}  \\ \hline
  W^1 & 0   & \cdots & 0 & 0  \\ 
  0   & W^2 & \cdots & 0 & 0 \\
  \vdots & \vdots & \ddots & \vdots & \vdots  \\
  0   & 0   & \cdots & W^\ell & 0 
 \end{array}\right]} := \bmat{ N_{ux} & N_{uw}  \\ N_{vx} & N_{vw} }.
\end{align*}
This decomposition of the NN, depicted in
Figure~\ref{fig:NNLFT}, isolates the nonlinearities $\phi$ in preparation for the stability analysis.

\vspace{-0.1in}
\begin{figure}[h]
  \centering
  \includegraphics[width=0.22\textwidth]{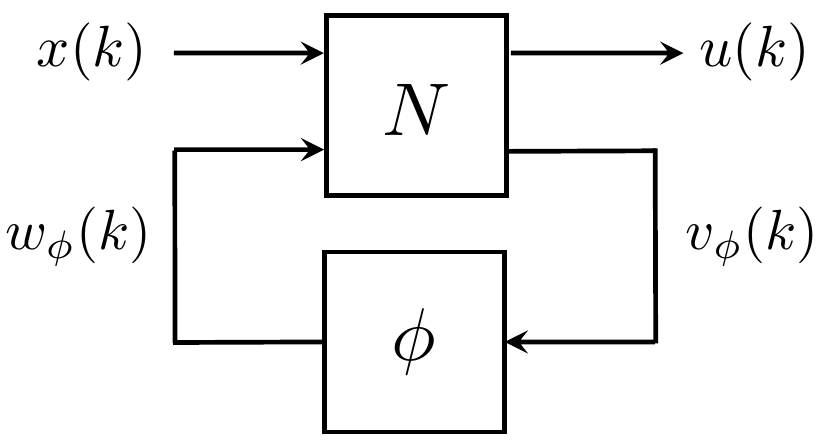}
  \caption{NN representation to isolate the nonlinearities $\phi$.}
  \label{fig:NNLFT}    
\end{figure}

\vspace{-0.2cm}
\subsection{Quadratic Constraints: Scalar Activation Functions}

The stability analysis relies on quadratic constraints (QCs) to bound 
the activation function.  
A typical
constraint is the sector bound as defined next.
\begin{definition}
  \label{def:GlobalSector}
  Let $\alpha \le \beta$ be given.  The function
  $\varphi: \R \rightarrow \R$ lies in the (global) sector
  $[\alpha,\beta]$ if:
  \begin{align}
    ( \varphi(\nu) - \alpha \nu ) \cdot
       (\beta \nu - \varphi(\nu)) \ge 0
    \,\,\, \forall \nu \in \R.
  \end{align}
\end{definition}
The interpretation of the sector $[\alpha,\beta]$ is  that
$\varphi$ lies between lines passing through the origin with slope
$\alpha$ and $\beta$.  Many activation functions are bounded in the
sector $[0,1]$, e.g. $\tanh$ and ReLU. 
Figure~\ref{fig:tanh} illustrates $\varphi(\nu) = \tanh(\nu)$ (blue solid)
and the global sector defined by $[0,1]$ (red solid lines).
\vspace{-0.1in}
\begin{figure}[h]
  \centering
  \includegraphics[width=0.42\textwidth]{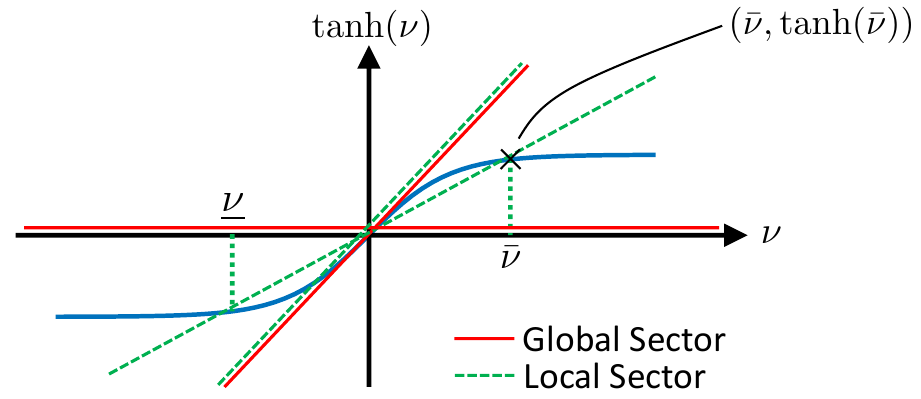}
  \caption{Sector constraints on $\tanh$}
  \label{fig:tanh}    
\end{figure}
\vspace{-0.1in}

The global sector is often too coarse for
analysis; thus, we consider a local sector  for
tighter bounds.

\begin{definition}
  \label{def:LocalSector}
  Let $\alpha$, $\beta$, $\underline \nu$, $\bar \nu \in\R$ 
  with $\alpha \le \beta$ and $\underline \nu \le 0 \le \bar \nu$.
  The function $\varphi: \R \rightarrow \R$ satisfies the local sector
  $[\alpha,\beta]$ if
    $( \varphi(\nu) - \alpha \, \nu ) \cdot
     (\beta \, \nu - \varphi(\nu)) \ge 0
    \,\,\, \forall \nu \in [\underline{\nu},\bar{\nu} ].$
\end{definition}

As an example,  $\varphi(\nu):=\tanh(\nu)$ restricted to the
interval $[-\bar{\nu},\bar{\nu}]$  satisfies the local
sector bound $[\alpha,\beta]$ with
$\alpha:=\tanh(\bar\nu)/\bar{\nu}>0$ and
$\beta:=1$. As shown in Figure~\ref{fig:tanh}  (green
dashed lines), the local sector provides a tighter bound than the global sector.  These bounds are
valid for a symmetric interval around the origin with
$\underline\nu = -\bar\nu$; non-symmetric intervals ($\underline\nu \ne -\bar \nu$) can be handled similarly.

\subsection{Quadratic Constraints: Combined Activation Functions}

Local sector constraints can also be defined for the combined
nonlinearity $\phi$, given by \eqref{eq:phi_def}.  Let
$\underline{v}, \bar{v} \in \R^{n_\phi}$ be given with
$\underline{v} \leq 0 \leq \bar{v}$. Suppose that the activation input
$v_\phi \in \R^{n_\phi}$ lies element-wise in the interval
$[\underline{v},\bar{v}]$ and the $i^{th}$ input/output pair is
$w_i = \varphi( v_i )$, and the scalar activation function
satisfies the local sector $[\alpha_i,\beta_i]$ with the input restricted to
$v_i \in [\underline{v}_i,\bar{v}_i]$ for $i=1,\ldots,n_\phi$.  The
local sector bounds can be computed for $\varphi$ on the given
interval analytically (as above for $\tanh$). These local sectors can  be stacked into vectors
$\alpha_\phi, \beta_\phi \in \R^{n_\phi}$ that provide 
QCs satisfied by the
combined nonlinearity $\phi$.
 
\begin{lemma} \label{lemma:sectorQC} Let $\alpha_\phi$, $\beta_\phi$,
  $\underline v$, $\bar v \in\R^{n_\phi}$ be given with
  $\alpha_\phi \le \beta_\phi$, and $\underline v \le 0 \leq \bar v$. Suppose that $\phi$ satisfies the local sector
  $[\alpha_\phi,\beta_\phi]$ element-wise for
  all $v_\phi \in [\underline v,\bar v]$. For any $\lambda \in \R^{n_\phi}$
  with $\lambda \ge 0$, and for all $v_\phi \in [\underline v,\bar v], 
  \, w_\phi = \phi( v_\phi )$, it holds 
  \begin{align}
    & {\small\bmat{v_\phi \\ w_\phi}^\top 
    \bmat{-2A_\phi  B_\phi \Lambda & (A_\phi + B_\phi)\Lambda \\ (A_\phi + B_\phi)\Lambda & -2 \Lambda}
    \bmat{v_\phi \\ w_\phi} \ge 0 }, \label{eq:local_QC}
  \end{align}
  where $A_\phi = diag(\alpha_\phi), \ B_\phi = diag(\beta_\phi)$, and $\Lambda = diag(\lambda).$
\end{lemma}
\begin{proof}
The proof can be found in \cite[Lemma 1]{Yinstabanaly}.
\end{proof}


In order to apply the local sector bounds in the stability analysis, we must first compute the bounds $\underline{v}, \overline{v} \in \R^{n_\phi}$ on the activation input $v_\phi$. The process to compute the bounds is briefly discussed here with more details provided in \cite{Gowal2018, Yinstabanaly}. The first step is to find the smallest hypercube that bounds the state constraint set: $X \subseteq  \{x: \underline{x} \leq x \leq \overline{x}\}$. Therefore, $w^0$ (defined in \eqref{eq:def_w0}) is bounded by $\underline{w}^0=\underline{x}$ and $\overline{w}^0 = \overline{x}$. Define $c = \frac{1}{2}(\overline{w}^0 + \underline{w}^0)$, $r = \frac{1}{2}(\overline{w}^0 - \underline{w}^0)$, and denote $y^\top$ as the $i^{th}$ row of $W^1$. Then the first activation input $v^1 = W^1 w^0$ is bounded by $[\underline{v}^1, \overline{v}^1]$, where 
$\overline{v}_i^1 = y^\top c+ \sum_{j=1}^{n_0}|y_j r_j|, \ \text{and} \ \ \underline{v}_i^1 = y^\top c- \sum_{j=1}^{n_0}|y_j r_j|.$
If the activation functions $\phi^1$ are monotonically non-decreasing, then the first activation output $w^1$ is bounded by $\underline{w}^1 = \phi^1(\underline{v}^1)$ and $\overline{w}^1 = \phi^1(\overline{v}^1)$. This process can be propagated through all layers of the NN to obtain the bounds $\underline{v}, \overline{v} \in \R^{n_\phi}$ for the activation input $v_\phi$. The remainder of the paper will assume the local sector bounds have been computed as briefly summarized in the following property.
\begin{property}
  \label{prop:NNsector}  
  Let the state constraint set $X$ and its corresponding activation input bounds $\underline{v}, \overline{v}$ be given.  There exist
  $\alpha_\phi$, $\beta_\phi$ such that $\phi$
  satisfies the local sector for all $v_\phi \in [\underline{v}, \overline{v}]$.
\end{property}

\subsection{Lyapunov Condition}

This section uses a Lyapunov function and the local sector to compute an inner approximation for
the ROA of the feedback system of $G$ and $\pi$.  

\begin{theorem}(\hspace{0.01in}\cite{Yinstabanaly})
  \label{thm:NominalLyap}
  Consider the feedback system of plant $G$ in \eqref{eq:NominalSys}
  and NN $\pi$ in \eqref{eq:NNlong} with equilibrium point
  $x_*=0_{n_G \times 1}$, and the state constraint set $X$. Let
  $\bar{v}, \underline{v},
  \alpha_\phi, \beta_\phi \in \R^{n_\phi}$ be given vectors
  satisfying Property~\ref{prop:NNsector} for the NN and the set $X$.  Denote the
  $i^{th}$ row of the matrix $H$ by $H_i^\top$ and define 
  \begin{align*}
    R_V :={\small\bmat{I_{n_G} & 0_{n_G \times n_\phi} \\ N_{ux} & N_{uw}} }, 
    \, \mbox{ and } \,
    R_\phi := {\small\bmat{N_{vx} & N_{vw} \\ 0_{n_\phi \times n_G} & I_{n_\phi}}}.
  \end{align*}
  If there exists a matrix $P \in \mathbb{S}_{++}^{n_G}$, and vector
  $\lambda \in \R^{n_\phi}$ with $\lambda \ge 0$ such that $\Lambda = diag(\lambda)$ satisfying
  \begin{subequations}
  \begin{align} 
    \nonumber
    & R_{V}^\top \bmat{A_G^\top P A_G - P & A_G^\top P B_G 
     \\ B_G^\top P A_G & B_G^\top P B_G} R_{V}  \\
    & \hspace{0.3in}
     +  R_{\phi}^\top \bmat{-2A_\phi  B_\phi \Lambda & (A_\phi + B_\phi)\Lambda \\ (A_\phi + B_\phi)\Lambda & -2 \Lambda} R_{\phi} < 0, 
    \label{eq:diss_nominal} \\
    &{\small\bmat{h_i^2 & H_i^\top \\ H_i & P}  }
    \ge 0, \,\,\,i = 1,\cdots, n_X, 
    \label{eq:setcontain}
  \end{align}
  \end{subequations}
   then: (i) the feedback system consisting of $G$ and $\pi$ is locally asymptotically
  stable around $x_*$, and (ii) the set $\mathcal{E}(P)$, defined by \eqref{eq:epsil_def}, is an 
  inner-approximation to the ROA.
\end{theorem}
\begin{proof}
The proof can be found in \cite[Theorem 1]{Yinstabanaly}.
\end{proof}

The Lyapunov condition~\eqref{eq:diss_nominal} is convex
in $P$ and $\lambda$ if the weight matrix $N$ is given, and thus we can efficiently compute the ROA inner-estimates. However, this condition is computationally intractable for NN controller synthesis, as it is nonconvex  if we search over $N$, $P$, and $\lambda$ simultaneously.

\section{Convex Stability and Safety Conditions}\label{sec:convex_condition}
In \cite{Pauli2020, revay2020convex}, $\alpha_\phi$ is set to zero to formulate convex constraints. However, this restriction is too coarse for stability analysis of NN controlled systems. Instead, we perform a loop transformation as shown in Fig.~\ref{fig:looptran} to convexify the stability condition without having restrictions on $\alpha_\phi$ and $\beta_\phi$.
\vspace{-0.2in}
\begin{figure}[h]
  \centering
  \includegraphics[width=0.42\textwidth]{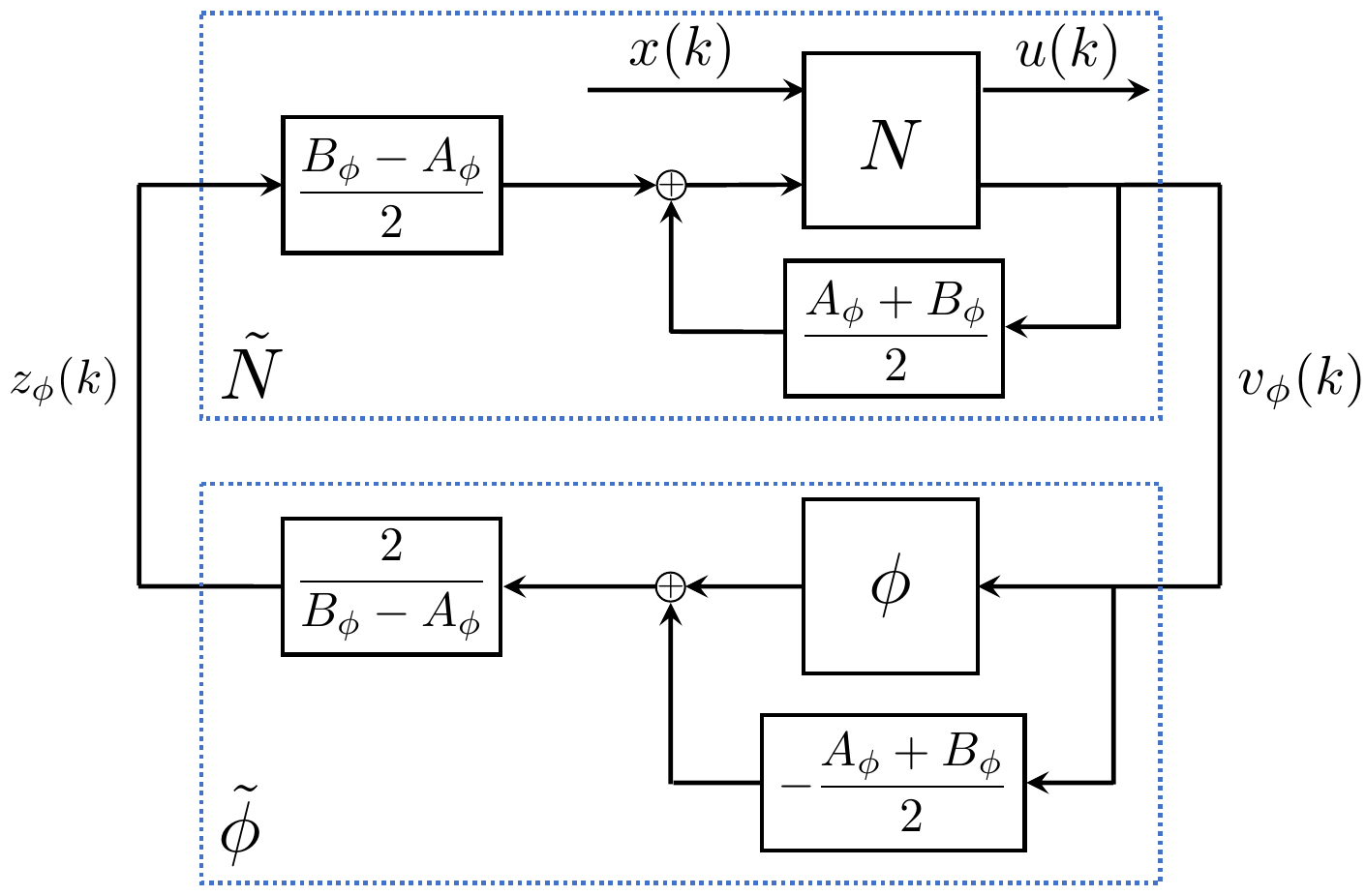}
  \caption{Loop transformation. If $\phi$ is in the sector $[\alpha_\phi, \beta_\phi]$, then $\tilde{\phi}$ is in the sector $[-1_{n_\phi \times 1}, 1_{n_\phi\times 1}]$. }
  \label{fig:looptran}    
\end{figure}

\subsection{Loop transformation}
Through loop transformation, we obtain a new representation of the NN controller, which is equivalent to \eqref{eq:NNlft},
\begin{subequations} \label{eq:NN_represent2}
  \begingroup
\allowdisplaybreaks
\begin{align}
    \begin{bmatrix}u(k)\\ v_\phi(k)\end{bmatrix} & = \tilde{N}  \begin{bmatrix}x(k) \\ z_\phi(k)\end{bmatrix}, \label{eq:NN_Ntilde}\\
    z_\phi(k) &= \tilde{\phi}(v_\phi(k)),
\end{align}
\endgroup
\end{subequations}
where $\tilde{N}$ and $\tilde{\phi}$ are defined in Fig.~\ref{fig:looptran}. Here, we also partition $\tilde{N}$ 
compatibly with the inputs $(x,z_\phi)$ and outputs $(u,v_\phi)$ 
\begin{align}
    \tilde{N} = {\small\bmat{\tilde{N}_{ux} & \tilde{N}_{uz} \\ \tilde{N}_{vx} & \tilde{N}_{vz}} }.
\end{align}

The loop transformation normalizes the nonlinearity $\tilde{\phi}$ to lie in the sector $[-1_{n_\phi \times 1}, 1_{n_\phi \times 1}]$. As a result, $\tilde{\phi}$ satisfies the sector QC 
for any $\Lambda=diag(\lambda)$ with $\lambda \ge 0$:
\begin{align}
    {\small\bmat{v_\phi \\ z_\phi}^\top \bmat{\Lambda & 0 \\ 0 & -\Lambda} \bmat{v_\phi \\ z_\phi} \ge 0}, \ \forall v_\phi \in [\underline{v}, \overline{v}]. \label{eq:shifted_local_QC}
\end{align}
The input to $N$ is transformed by the following  equation:
\begin{align}
w_\phi(k) = \frac{B_\phi - A_\phi}{2} z_\phi(k) + \frac{A_\phi + B_\phi}{2} v_\phi(k).  \label{eq:input_Ntilde}
\end{align}
The transformed matrix $\tilde{N}$ can be computed by combining this relation with \eqref{eq:def_N}. Substituting \eqref{eq:input_Ntilde} into \eqref{eq:def_N} we obtain 
\begingroup
\allowdisplaybreaks
\begin{align}
    u(k) &= N_{ux} x(k) + C_1 z_\phi(k) + C_2 v_\phi(k), \label{eq:u_1}\\
    v_\phi(k) &= N_{vx} x(k) + C_3 z_\phi(k) + C_4 v_\phi(k), \label{eq:vphi_1} \\
\text{where} \ \ \ \ \ C_1 &= N_{uw}\frac{B_\phi - A_\phi}{2}, \  C_2 = N_{uw}\frac{A_\phi + B_\phi}{2}, \nonumber \\
C_3 &= N_{vw}\frac{B_\phi - A_\phi}{2}, \ C_4 = N_{vw}\frac{A_\phi + B_\phi}{2}. \nonumber
\end{align}
\endgroup
The expression for $v_\phi(k)$ can be solved from \eqref{eq:vphi_1}:
\begin{align}
    v_\phi(k) = (I - C_4)^{-1} N_{vx} x(k) + (I - C_4)^{-1} C_3 z_\phi(k). \label{eq:vphi_2}
\end{align}
Substituting \eqref{eq:vphi_2} into \eqref{eq:u_1} yields
\begingroup
\allowdisplaybreaks
\begin{align}
    u(k) &= (N_{ux}+C_2(I-C_4)^{-1}N_{vx}) x(k) \nonumber \\ 
    &~~~~~~~~~~~~~~~ + (C_1+C_2(I-C_4)^{-1}C_3) z_\phi(k).\label{eq:u_2} 
\end{align}
\endgroup
Matching \eqref{eq:vphi_2} and \eqref{eq:u_2} with \eqref{eq:NN_Ntilde} we get 
\begin{align*}
    &\tilde{N} = {\small\bmat{N_{ux}+C_2(I-C_4)^{-1}N_{vx} & C_1+C_2(I-C_4)^{-1}C_3 \\ (I - C_4)^{-1}N_{vx} & (I - C_4)^{-1}C_3}.}
\end{align*}
Thus, $\tilde{N}$ is a function of $N$ denoted concisely as $\tilde{N}=f(N)$. It is important to note that $\tilde{N}$ depends on $N$ both directly, and also indirectly through its dependence on $(A_\phi,B_\phi$).  Specifically, suppose both $N$ and a hypercube state bound $[\underline{x},\bar{x}]$ are given.  Then $\tilde{N}$ is constructed by: (i) propagating $[\underline{x},\bar{x}]$ through the NN to compute bounds $[\overline{v}, \underline{v}]$ on the activation inputs, (ii) computing local sector bounds $(A_\phi,B_\phi)$ consistent with $[\overline{v}, \underline{v}]$, and (iii) performing the steps in this section to compute $\tilde{N}$ from $(N,A_\phi,B_\phi)$.  

\vspace{-0.1in}
\subsection{Stability condition after loop transformation}
 Similar to the original Lyapunov condition~\eqref{eq:diss_nominal}, the new Lyapunov condition for the feedback system of $G$ in \eqref{eq:NominalSys} and NN in \eqref{eq:NN_represent2} can be written as
   \begingroup
\allowdisplaybreaks
\begin{align}
    &\tilde{R}_{V}^\top \left[\begin{smallmatrix}A_G^\top P A_G - P & A_G^\top P B_G 
     \\ B_G^\top P A_G & B_G^\top P B_G\end{smallmatrix}\right] \tilde{R}_{V}  +  \tilde{R}_{\phi}^\top \left[\begin{smallmatrix} \Lambda & 0 \\ 0 & -\Lambda\end{smallmatrix} \right]\tilde{R}_{\phi} < 0, \label{eq:new_lyap} \\
     &\text{where} \ \tilde{R}_{V} = {\small\begin{bmatrix} I_{n_G} & 0 \\ \tilde{N}_{ux} & \tilde{N}_{uz} \end{bmatrix}}, \ \text{and} \ \tilde{R}_{\phi} = {\small\begin{bmatrix} \tilde{N}_{vx} & \tilde{N}_{vz} \\ 0 & I_{n_\phi} \end{bmatrix}.}\label{eq:RV_Rphi}
\end{align}
\endgroup
\vspace{-0.15in}
\begin{lemma}
Consider the feedback system of $G$ in \eqref{eq:NominalSys} and NN in \eqref{eq:NNlft} with the state constraint set $X$. If there exist a matrix $P \in \mathbb{S}_{++}^{n_G}$, and vector $\lambda \in \R^{n_\phi}$ with $\lambda \ge 0$ such that \eqref{eq:new_lyap} (where $\tilde{N} = f(N)$) and \eqref{eq:setcontain} hold, then: (i) the feedback system consisting of $G$ in \eqref{eq:NominalSys} and NN in \eqref{eq:NNlft} is locally asymptotically
  stable around $x_*$, and (ii) the set $\mathcal{E}(P)$ is a ROA inner-approximation for it.
\end{lemma}
\begin{proof}
It follows from the assumption that \eqref{eq:new_lyap} and \eqref{eq:setcontain} hold that the  feedback system of $G$ in \eqref{eq:NominalSys} and NN in \eqref{eq:NN_represent2} is locally asymptotically stable around $x_*$, and $\mathcal{E}(P)$ is its ROA inner-approximation. Since the representations \eqref{eq:NNlft} and \eqref{eq:NN_represent2} of NN are equivalent, the feedback system of $G$ in \eqref{eq:NominalSys} and NN in \eqref{eq:NNlft} is identical to the feedback system of $G$ in \eqref{eq:NominalSys} and NN in \eqref{eq:NN_represent2}. As a result, the feedback system consisting of $G$ in \eqref{eq:NominalSys} and NN in \eqref{eq:NNlft} is locally asymptotically stable around $x_*$, and the set $\mathcal{E}(P)$ is a ROA inner-approximation for it.
\end{proof}

The new Lyapunov condition~\eqref{eq:new_lyap} is convex in $P$ and $\Lambda$ using $\tilde{N} = f(N)$, where $N$ is given. To incorporate the stability condition in the IL process, we will proceed by treating $\tilde{N} \in \R^{(n_u+n_\phi)\times(n_G+n_\phi)}$ as a decision variable along with $P$ and $\Lambda$, and try to derive a stability condition that is convex in $(P, \Lambda, \tilde{N})$. Substitute \eqref{eq:RV_Rphi} into \eqref{eq:new_lyap} to obtain
{\small
\begin{align}
    &\bigg[ \star \bigg]^\top \begin{bmatrix} P & 0 \\ 0 & \Lambda \end{bmatrix}\begin{bmatrix}A_G + B_G \tilde{N}_{ux} & B_G \tilde{N}_{uz} \\ \tilde{N}_{vx} & \tilde{N}_{vz}  \end{bmatrix}  - \begin{bmatrix} P & 0 \\ 0 & \Lambda\end{bmatrix} < 0. \nonumber
\end{align}}
Applying Schur complements yields the equivalent condition
{\small
\begin{align}
        {\small\begin{bmatrix} P & 0  & A_G^\top + \tilde{N}_{ux}^\top B_G^\top  & \tilde{N}_{vx}^\top \\ 0 & \Lambda & \tilde{N}_{uz}^\top B_G^\top & \tilde{N}_{vz}^\top  \\ A_G + B_G \tilde{N}_{ux} & B_G \tilde{N}_{uz} & P^{-1} & 0 \\\tilde{N}_{vx} & \tilde{N}_{vz} & 0 & \Lambda^{-1}\end{bmatrix}} > 0, \label{eq:big_lyap1}
\end{align}
}
\hspace{-0.14in} and $P > 0$, $\Lambda >0$. Now \eqref{eq:big_lyap1} is linear in $\tilde{N}$, but still nonconvex in $P$ and $\Lambda$. Multiplying  \eqref{eq:big_lyap1} on the left and right by $diag(P^{-1},\Lambda^{-1},I_{n_G},I_{n_\phi})$ we obtain 
{\small
\begin{align}
         {\small\begin{bmatrix} Q_1 & 0  & Q_1 A_G^\top + L_1^\top B_G^\top  & L_3^\top \\ 0 & Q_2 & L_2^\top B_G^\top & L_4^\top  \\ A_G Q_1 + B_G L_1 & B_G L_2 & Q_1 & 0 \\L_3 & L_4 & 0 & Q_2 \end{bmatrix}} > 0, \label{eq:big_lyap2} 
\end{align}}
\hspace{-0.12in} where $Q_1 = P^{-1} > 0$, $Q_2 = \Lambda^{-1} > 0$, 
$L_1 = \tilde{N}_{ux} Q_1$, $L_2 = \tilde{N}_{uz} Q_2$, $L_3 = \tilde{N}_{vx} Q_1$, and $L_4 = \tilde{N}_{vz} Q_2$. 

The stability condition~\eqref{eq:big_lyap2} is convex in the decision variables   $(Q_1, Q_2, L_1,\hdots, L_4)$, where $Q_1 \in \mathbb{S}_{++}^{n_G}$, $Q_2 \in \mathbb{S}_{++}^{n_\phi}$ and $Q_2$ is a diagonal matrix, $L_1 \in \R^{n_u \times n_G}$, $L_2 \in \R^{n_u \times n_\phi}$, $L_3 \in \R^{n_\phi \times n_G}$, and $L_4 \in \R^{n_\phi \times n_\phi}$. Variables $(P, \Lambda, \tilde{N})$ that satisfy the Lyapunov condition~\eqref{eq:new_lyap} can be recovered using the computed $(Q_1, Q_2, L_1,\hdots, L_4)$ through the following equations: $P = Q_1^{-1}, \Lambda = Q_2^{-1}$, and
\begin{align}
    &\tilde{N} = L Q^{-1}, \label{eq:Ntilde_LQ}\\
    &\text{where} \ Q := \left[\begin{smallmatrix}Q_1 & 0 \\ 0 & Q_2\end{smallmatrix}\right], \ \text{and} \ L := \left[\begin{smallmatrix}L_1 & L_2 \\ L_3 & L_4 \end{smallmatrix}\right].\label{eq:QL_def}
\end{align}
Thus, the convex stability condition~\eqref{eq:big_lyap2} allows us to search over $P$, $\Lambda$, and $\tilde{N}$ simultaneously.

Moreover, to enforce the safety condition $(x(k) \in X \ \forall k \ge 0)$ of the system, convex constraints on $Q_1$ are imposed:
\begin{align}
    H_i^\top Q_1 H_i \leq h_i^2, 
      \ i = 1,\cdots,n_X, \label{eq:setcontain_Q1}
\end{align}
which is derived directly from \eqref{eq:setcontain} by Schur complements, and using $Q_1 = P^{-1}$. Again, this ensures $\mathcal{E}(Q_1^{-1}) \subseteq X$.

Denote the LMIs \eqref{eq:big_lyap2}, \eqref{eq:setcontain_Q1} with $Q_1 >0$ and $Q_2 > 0$  altogether as LMI$(Q, L) > 0$, which will later be incorporated in the IL process to learn robust NN controllers.

\begin{remark}
Note that model uncertainties are not considered in the paper. They can be incorporated in \eqref{eq:diss_nominal} or \eqref{eq:new_lyap} using integral quadratic constraints (IQCs) as in \cite{Yinstabanaly}. However, to derive convex  conditions, like \eqref{eq:big_lyap2}, only limited types of uncertainties may be incorporated. This in turn will allow for Zames-Falb IQCs to refine the description of activation functions by capturing their slope restrictions.
\end{remark}

\vspace{-0.1cm}
\section{Safe Imitation Learning Algorithm}\label{sec:safeLearn}
Given  state and  control data  pairs from  the  expert demonstrations, we use a loss function $\mathcal{L}(N)$ to train NN controllers with weights $N$ to match the data. Common choices of the loss function include mean squared error, absolute error, and cross-entropy. In general, $\mathcal{L}(N)$ is non-convex in $N$. We propose the next optimization to solve the safe IL problem, 
\vspace{-0.2cm}
\begin{subequations}\label{eq:AMDD_opt}
\begin{align}
    \min_{N, Q, L} \ &\eta_1 \mathcal{L}(N) - \eta_2 \log \det (Q_1) \label{eq:cost_fcn}\\
\text{s.t.} \ &\text{LMI}(Q, L) > 0 \label{eq:stab_safe_constr}\\
&f(N)Q = L \label{eq:connect_constr} 
\end{align}
\end{subequations}
where $Q$ and $L$ are defined in \eqref{eq:QL_def}.
The optimization has separate objectives. The cost function~\eqref{eq:cost_fcn} combines the IL loss function with a term that attempts to increase the volume of $\mathcal{E}(Q_1^{-1})$ (which is proportional to $\det(Q_1)$ ). The parameters $\eta_1, \eta_2 > 0$ reflect the relative importance between imitation learning accuracy and size of the robustness margin. The optimization has two sets of decision variables: $N$ and $(Q, L)$. The former is involved in mimicking the expert behaviour, and the latter are involved in the stability and safety constraints \eqref{eq:stab_safe_constr}. The two sets of variables are connected through the equality constraint \eqref{eq:connect_constr}. Note that \eqref{eq:connect_constr} is equivalent to $f(N)= L Q^{-1}$, and the term on the right-hand side equals to $\tilde{N}$ from \eqref{eq:Ntilde_LQ}. Therefore, \eqref{eq:connect_constr} essentially means that the first set of decision variable $N$, after being transformed by the nonlinear function $f$, satisfies the stability and safety constraints. 

Similar to \cite{Pauli2020}, we use the alternating direction method of multipliers (ADMM) algorithm to solve this constrained learning problem. We first define an augmented loss function
\begin{align}
    &\mathcal{L}_a(N, Q, L, Y) = \ \eta_1 \mathcal{L}(N) - \eta_2 \log \det (Q_1) \nonumber \\
    &~~~~~+ \text{tr}\left(Y^\top \cdot \left(f(N)Q - L \right)\right)  + \frac{\rho}{2} \norm{f(N) Q -L }^2_F,
\end{align}
where $\norm{\cdot}_F$ is the Frobenius norm, $Y \in \R^{(n_u+n_\phi) \times (n_G + n_\phi)}$ is the Lagrange multiplier, and $\rho > 0$ is the regularization parameter typically affecting the convergence rate of ADMM. The ADMM algorithm takes the following form:

\noindent 1. $N$-update: $N^{k+1} = \arg \min_N \mathcal{L}_a(N, Q^k, L^k, Y^k)$.

\noindent 
\begin{align}
\text{2. $(Q,L)$-update: }&(Q, L)^{k+1} = \arg \min_{Q, L}\mathcal{L}_a(N^{k+1}, Q, L, Y^k) \nonumber \\
& ~~~~~~~~~~~~~~~~~~~~~ \text{s.t.} \ \text{LMI}(Q, L) > 0 \nonumber
\end{align}
\noindent 3. $Y$-update: If $\norm{f(N^{k+1})Q^{k+1}-L^{k+1}}_F \leq \sigma$, where $\sigma > 0$ is the stopping tolerance, then the algorithm has converged, and we have found a solution to \eqref{eq:AMDD_opt}, so terminate the algorithm. Otherwise, update $Y$ and return to step 1.
\begin{align}
&Y^{k+1} = Y^k + \rho \left( f(N^{k+1}) Q^{k+1}  -L^{k+1} \right)
\nonumber
\end{align}
The unconstrained optimization in Step 1 can be solved using gradient based algorithm. The optimization in Step 2 is convex, and can be solved effectively using SDP solvers. The variable $Y$ in Step 3 accumulates the deviation from the constraint~\eqref{eq:connect_constr}. 
The individual steps are well-behaved: Gradient descent in Step 1 almost surely converges to local optima \cite{Jason2016GD} and Step 2 always obtains a global optimum.  However, the loss $\mathcal{L}(N)$ and constraint \eqref{eq:connect_constr} are nonconvex and hence the proposed ADMM does not guarantee convergence to a global optima. Convergence properties of ADMM without convexity are still subject to ongoing research \cite{wang2019global}. However, any converged solution provides a safe NN controller with stability and safety guarantees. 

In Step 2, $Q$ and $L$ introduce $\mathcal{O}(n_G + n_\phi)$ and $\mathcal{O}((n_G+n_\phi) \times (n_u+n_\phi))$ decision variables. The  complexity of primal/dual SDP solvers scales cubically with the number of decision variables. Hence Step 2 will be computationally expensive if the number of activation functions $n_\phi$ is large.


\vspace{-0.1cm}
\section{Examples}\label{sec:example}
In the following examples,  Step~1 in the ADMM algorithm is implemented on Tensorflow, and solved by ADAM~\cite{kingma2014adam}. Step~2 is formulated using CVX, and is solved by MOSEK. The mean squared error is chosen as the loss function $\mathcal{L}(N)$. The code is available at \href{https://github.com/heyinUCB/IQCbased_ImitationLearning}{https://github.com/heyinUCB/IQCbased\_ImitationLearning}

\subsection{Inverted pendulum}
Consider an inverted pendulum with mass $m = 0.15$ kg, length $l = 0.5$ m, and
friction coefficient $\mu = 0.5$ Nms$/$rad. The discretized and linearized dynamics are:
	\begin{align}
	\bmat{x_1(k+1) \\ x_2(k+1)} = \bmat{1 & \delta \\ \frac{g \delta}{l} & 1-\frac{\mu \delta}{ml^2}}\bmat{x_1(k) \\ x_2(k)} + \bmat{0\\ \frac{\delta}{ml^2}}u(k), \nonumber
	\end{align}
 where the states $x_1$, $x_2$ represent the angular position (rad) and velocity (rad/s), $u$ is the control
input (Nm), and $\delta = 0.02$~s is the sampling time. The state constraint set is $X = [-2.5, 2.5]\times [-6, 6]$. To generate state and control data pairs for IL, we design an explicit model predictive controller (MPC) to serve as the expert. By fitting a NN controller to the explicit MPC controller, we can expedite the evaluation of controllers during run-time \cite{MomoIL, karg2020efficient, SChen2018}. In this example, besides a NN controller, we will also provide its associated ROA inner-approximation that guarantees stability and safety. The NN controller is parameterized by a
2-layer, feedforward NN with $n_1 = n_2 = 10$ and $\tanh$ as the
activation function for both layers. Take $\rho = 1$, $\eta_1 = 100$, and $\eta_2 = 5$. The ADMM algorithm is terminated after 16 iterations, and $\norm{f(N)-LQ^{-1}}_F = 0.17$.

In Fig.~\ref{fig:penROA}, the plot on the left shows the learned NN controller with a blue surface, and state and control data pairs from expert demonstrations with orange dots; the plot on the right shows the ROAs of the MPC controller and the NN controller with oranges dots and a blue ellipsoid, respectively. We can notice that the ROA of the NN controller is tightly contained by the
state constraint set $X$ (shown with a gray rectangle), which guarantees the safety of the system. 
\vspace{-0.2in}
\begin{figure}[h]
	\centering
	\includegraphics[width=0.45\textwidth]{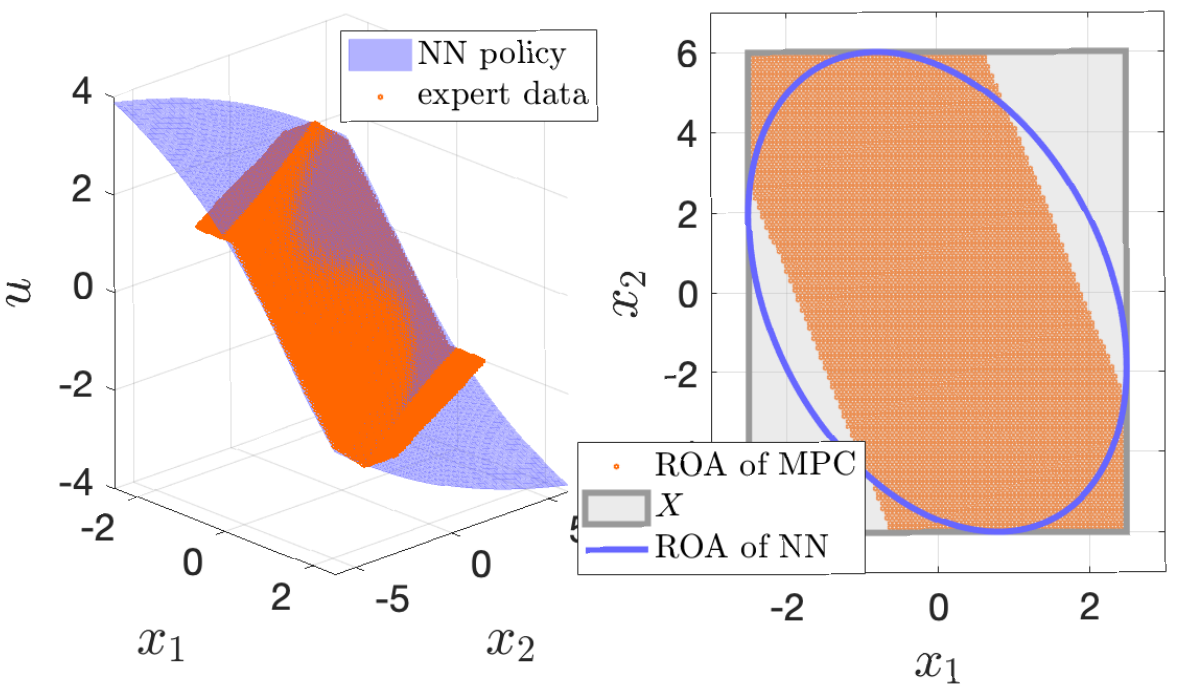}
	\caption{Left: NN controller vs. expert data from demonstrations; Right: ROAs of MPC controller and NN controller, and state constraint set $X$ of the inverted pendulum}
	\label{fig:penROA}    
\end{figure}

\subsection{Generic Transport Model}\label{sec:GTM_ex}
The Generic Transport Model (GTM) is a 5.5\% scale commercial aircraft. Linearizing and discretizing the longitudinal dynamics given in \cite{chakraborty2011nonlinear} with sampling time $\delta = 0.02$~s yields:
\begin{align}
\bmat{x_1(k+1) \\ x_2(k+1)} = \bmat{0.935 & 0.019 \\ -0.907 & 0.913}\bmat{x_1(k) \\ x_2(k)} + \bmat{-0.006\\ -1.120}u(k), \nonumber
\end{align}
where the states $x_1, x_2$ represent angle of attack (rad), and pitch rate (rad/s), and the control $u$ represents the elevator deflection (rad). Take the state constraint set as $X = [-2, 2] \times [-3, 3]$. In this example, we design an LQR controller to produce expert data. The NN controller is again parameterized by a 2-layer, feedforward NN with $n_1 = n_2 = 16$ and $\tanh$ as the
activation function for both layers. In this example, we will show how the parameter $\eta_2$ affects the result. To do so, two experiments are carried out using two sets of parameters $(\rho=1, \eta_1 = 100, \eta_2 = 5)$ and $(\rho=1, \eta_1 = 100, \eta_2 = 20)$, meaning that we care more about the size of the ROA inner-approximation, and less about the IL accuracy in the second experiment than we do in the first experiment. In both experiments, the ADMM algorithm is terminated in 20 iterations. 

The ROA inner-approximations of the NN controllers from the two experiments are shown in Figure \ref{fig:GTMROA}. The one computed with $\eta_2 = 5$  is shown with a magenta ellipsoid, and the one computed with $\eta_2 = 20$ is shown with a blue ellipsoid. First, it is important to notice that both NN controllers' ROA inner-approximations are larger than that of the expert's LQR controller (shown with a dashed gray ellipsoid), thanks to the second term in the cost function~\eqref{eq:cost_fcn}, which enhances the robustness of IL. Also, as expected, the ROA inner-approximation of the NN controller with $\eta_2 = 20$  is larger than that with $\eta_2 = 5$, since a larger $\eta_2$ leads to a larger ROA inner-approximation. However, the larger ROA inner-approximation comes at the cost of less accurate regression to the expert data. As shown in Figure~\ref{fig:GTMNN}, the mesh plot of the NN controller with $\eta_2 = 20$ (shown with a blue surface) is more off from the expert data (shown with orange stars) than that with $\eta_2 = 5$ (shown with a magenta surface).

\begin{figure}[h]
	\centering
	\includegraphics[width=0.3\textwidth]{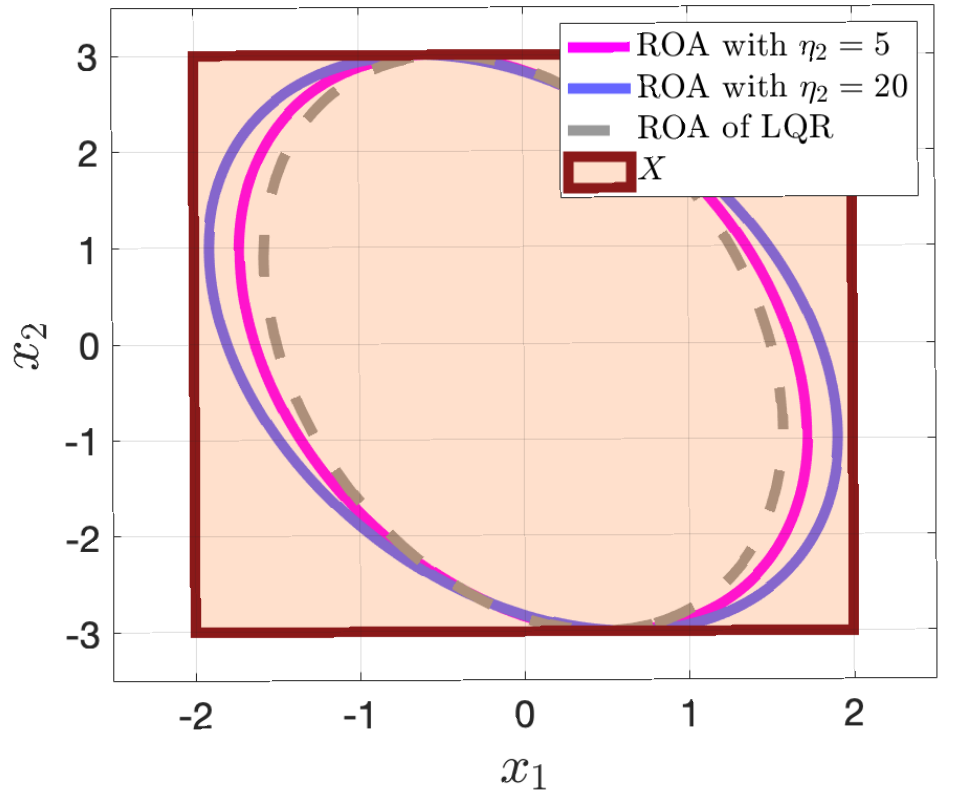}
	\caption{ROAs and state constraint set $X$ of GTM}
	\label{fig:GTMROA}    
\end{figure}

\vspace{-0.3cm}
\begin{figure}[h]
	\centering
	\includegraphics[width=0.35\textwidth]{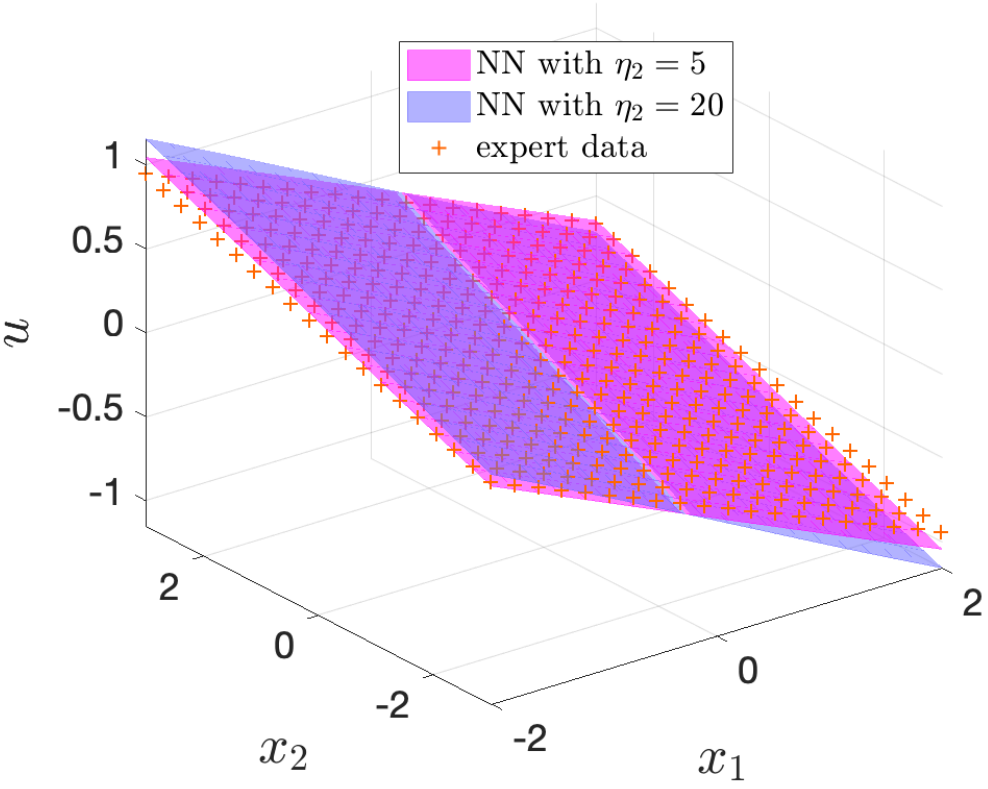}
	\caption{NN controllers vs. expert data of GTM}
	\label{fig:GTMNN}    
\end{figure}

\subsection{Vehicle Lateral Control}
Consider discretized vehicle lateral dynamics from \cite[Eqn 34]{Yinstabanaly} with sampling time $\delta=0.02$s, and a constant curvature $c \equiv 0$. Let $x = [e, \dot{e},e_\theta,\dot{e}_\theta]$ denote the plant state, where $e$ is the perpendicular distance to the lane edge (m),
and $e_\theta$ is the angle between the tangent to the straight section
of the road and the projection of the vehicle’s longitudinal
axis (rad). The
control $u$ is the steering angle of the front wheel (rad). The state constraint set is $X = [-2,2]\times[-5,5]\times[-1,1]\times[5,5]$. We design an explicit MPC law to serve as the expert, which is computed with an input constraint  $u(k) \in [-\pi/6, \pi/6]$ and a 5-step prediction horizon.  
The NN controller is parameterized by a
2-layer NN with $n_1 = n_2 = 10$ and $\tanh$ as the
activation function for both layers. We perform
two experiments using two sets of parameters
$(\rho = 1000, \eta_1 = 100, \eta_2 = 100)$ and $(\rho = 1000, \eta_1 = 100, \eta_2 = 500)$. In
both experiments, the ADMM algorithm is terminated in 20
iterations, and the achieved $\|f(N)-L Q^{-1}\|_F$ are 0.14 and 0.08.

As shown in Fig.~\ref{fig:veh_ROA}, the ROAs of the NN controllers  with $\eta_2=100$ and $\eta_2=500$ are contained by the state constraint set $X$ shown with maroon boxes, guaranteeing the safety of the system. As expected, the volume of the ROA with $\eta_2=500$ is larger than that with $\eta_2 = 100$, and the achieved det$(Q_1)$ for $\eta_2=100$ and $\eta_2=500$ are $543$ and $663$. However, the larger ROA comes at the cost of less accurate regression to
the expert data. As shown in Fig.~\ref{fig:veh_sim}, the simulated state and control signals with $\eta_2=100$ are closer to those of the MPC law than those with $\eta_2=500$. 

It takes $0.011$ and $0.012$ second to run the simulations for the NN controllers for $500$ time steps with $\eta_2=100$ and $\eta_2=500$ on a laptop with an Intel core i5 processor, and it takes $0.38$ second  for the explicit MPC law, which demonstrates the run-time advantage of NNs over explicit MPC laws.
\vspace{-0.2cm}
\begin{figure}[h]
	\centering
	\includegraphics[width=0.45\textwidth]{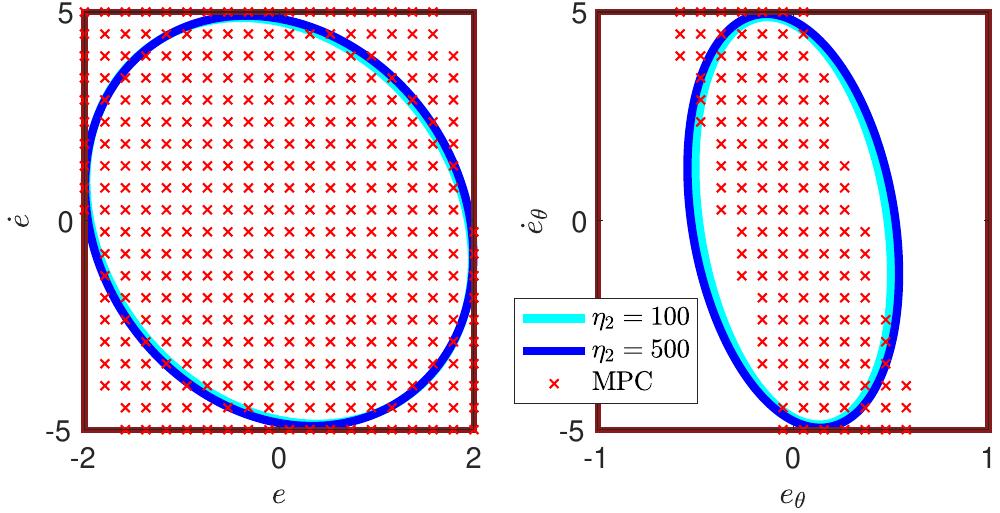}
	\vspace{-0.1cm}
	\caption{ROAs of the explicit MPC law (red crosses), and the NN controllers with $\eta_2=100$ (cyan curves) and $\eta_2 = 500$ (blue curves)}
	\label{fig:veh_ROA}    
\end{figure}
\begin{figure}[h]
	\centering
	\includegraphics[width=0.48\textwidth]{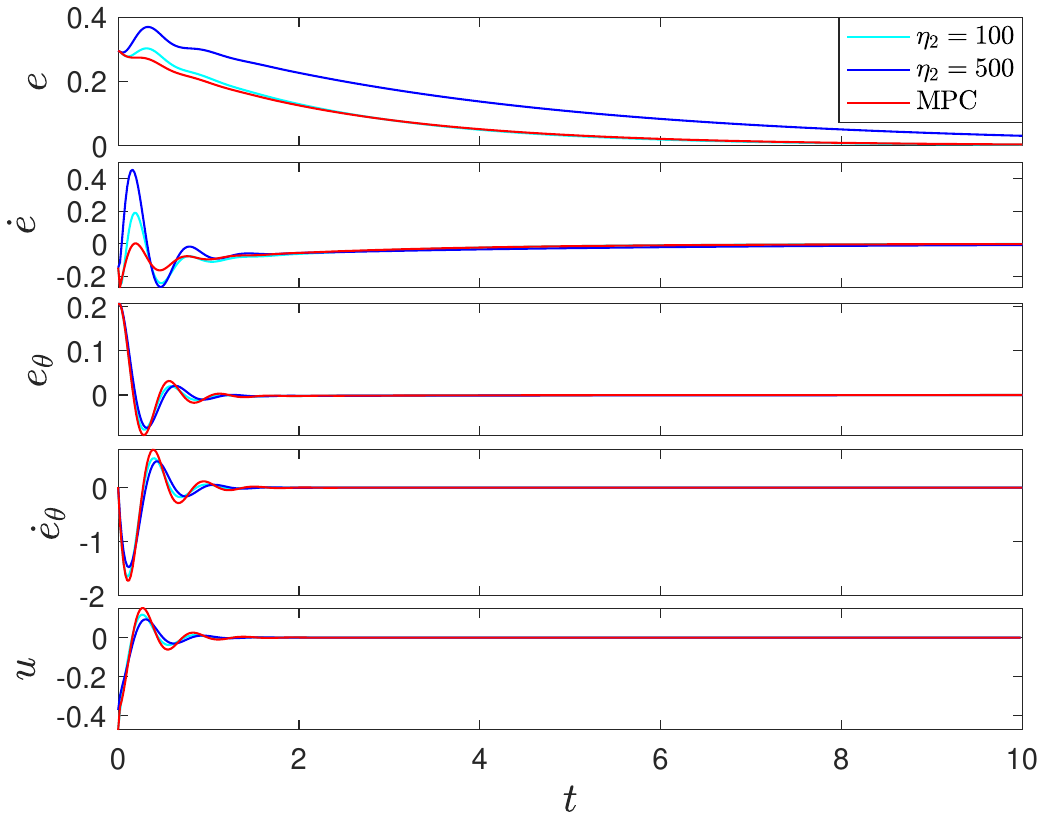}
	\vspace{-0.1cm}
	\caption{Simulated signals using the explicit MPC law (red curves), and the NN controllers with $\eta_2=100$ (cyan curves) and $\eta_2 = 500$ (blue curves)}
	\label{fig:veh_sim}    
\end{figure}

\section{Conclusions} \label{sec:conclusion} In this paper, we  present an IL algorithm with stability and safety guarantees. First, we derive convex stability and safety conditions for NN controlled systems. Then, we incorporate these conditions in the IL process, which trades off between the IL accuracy, and the size of the ROA. Finally, we propose an ADMM based algorithm to solve the safe IL problem. 

\bibliographystyle{IEEEtran}
\bibliography{IEEEabrv,bibfile}

\end{document}